\documentclass[conference]{IEEEtran}
\usepackage[top=0.75in, bottom=0.75in, left=0.7in, right=0.7in]{geometry}
\usepackage[pdftex]{graphicx}
\usepackage{float}
\usepackage{latexsym}
\usepackage[cmex10]{amsmath}
\usepackage{amssymb}
\usepackage{amsmath}
\usepackage{paralist}
\usepackage{leftidx}
\usepackage{mathrsfs}
\usepackage{multirow}
\usepackage{slashbox}
\usepackage{amsthm}
\usepackage{background}
\usepackage{algorithm}
\usepackage{algorithmic}

\usepackage[norule]{footmisc}
\IEEEoverridecommandlockouts

\theoremstyle{definition}

\theoremstyle{plain}

\theoremstyle{plain}

\theoremstyle{plain}

\theoremstyle{plain}
\newtheorem{lem}{Lemma}
\theoremstyle{remark}

\begin{document}
\SetBgContents{}
\title{Incentive Mechanism Design for Cache-Assisted D2D Communications: A Mobility-Aware Approach}
\author{\IEEEauthorblockN{Rui Wang$^*$, Jun Zhang$^*$ and K. B. Letaief$^{*\dag}$, \emph{Fellow, IEEE} }
	\IEEEauthorblockA{$^*$Dept. of ECE, The Hong Kong University of Science and Technology, $^\dag$Hamad Bin Khalifa University, Doha, Qatar\\
		Email: $^*$\{rwangae, eejzhang, eekhaled\}@ust.hk, $^\dag$kletaief@hbku.edu.qa} \thanks{This work was supported by the Hong Kong Research Grants Council under Grant No. 610113.}}
\maketitle
\begin{abstract}
Caching popular contents at mobile devices, assisted by device-to-device (D2D) communications, is considered as a promising technique for mobile content delivery. It can effectively reduce backhaul traffic and service cost, as well as improving the spectrum efficiency. However, due to the selfishness of mobile users, incentive mechanisms will be needed to motivate device caching. In this paper, we investigate incentive mechanism design in cache-assisted D2D networks, taking advantage of the user mobility information. An inter-contact model is adopted to capture the average time between two consecutive contacts of each device pair. A Stackelberg game is formulated, where each user plays as a follower aiming at maximizing its own utility and the mobile network operator (MNO) plays as a leader aiming at minimizing the cost. 
Assuming that user responses can be predicted by the MNO, a cost minimization problem is formulated. Since this problem is NP-hard, we reformulate it as a non-negative submodular maximization problem and develop a $(\frac{1}{4+\epsilon})$-approximation local search algorithm to solve it. In the simulation, we demonstrate that the local search algorithm provides near optimal performance. By comparing with other caching strategies, we validate the effectiveness of the proposed incentive-based mobility-aware caching strategy. 
\end{abstract}
\IEEEpeerreviewmaketitle

\section{Introduction}


To accommodate the exponentially growing mobile data traffic, especially the mobile video, wireless caching recently attracts lots of attentions \cite{poularakis2016code}
. By prefetching popular contents at wireless edges, i.e., local access points and mobile devices, mobile users can be served locally without connecting to the core network \cite{bastug2014living}. Accordingly, wireless caching is a promising technique to alleviate the backhaul burden, as well as lowering delays and service cost. Different from caching in wired networks, wireless caching enjoys a unique feature, i.e., the users are mobile. Recently, the user mobility information has been exploited to improve wireless caching. A general framework to design mobility-aware caching strategies was introduced in \cite{magazine}, along with a detailed discussion of different mobility models. Modeling the user mobility by a Markov chain model, a femto-caching strategy was proposed in \cite{poularakis2016code}. Taking the advantage of the temporal correlation of the user mobility, a mobility-aware caching strategy at user devices was proposed in \cite{mobilitycaching}.

Compared with the femto-caching systems, caching at user devices can facilitate device-to-device (D2D) communications. The user requests may be served via proximate D2D links without going through the base station (BS), which can significantly improve the spectrum efficiency and reduce the backhaul burden. However, there exists a new challenge in D2D caching networks. That is, considering the selfish nature of the mobile users, BSs cannot fully control the caching strategy at user devices. Thus, incentive mechanisms are needed to motivate the users to cache. 

To resolve the selfish issue in caching networks, some efforts have been made. In \cite{goemans2006market}, an incentive mechanism is designed to motivate caching at resident subscribers in an ad hoc network. Considering a peer-to-peer (P2P) system, an incentive scheme was proposed in \cite{P2P}, which rewards the peers based on the popularity of contents. However, different from conventional P2P systems, the topology of the D2D cache network is not fixed due to the user mobility. A recent work \cite{chen2016caching} considered a D2D caching network and proposed  an incentive mechanism based on a Stackelberg game, which unfortunately ignored the user mobility. Moreover, another crucial issue in previous studies on incentive mechanism in D2D caching networks is the high overhead incurred  by collecting global information at each user device. 

In this paper, we investigate the incentive mechanism design in cache-assisted D2D networks, while considering the user mobility. To take advantage of the user mobility pattern, an inter-contact model is considered, where the timeline of each pair of devices consists of \emph{contact times} and \emph{inter-contact times}. A pair of user devices are in contact when they are within the transmission range. Then, the \emph{contact times} are the times when two users are in contact, while the \emph{inter-contact times} are the times between two consecutive contact times.  Considering the selfish nature of the mobile users, the mobile network operator (MNO) should provide some incentives to motivate device caching. A Stackelberg game is formulated, where each user plays as a follower aiming at maximizing its utility by controlling its own cache storage, and the MNO plays as a leader aiming at minimizing the overall cost, including the service cost and the payment to users, by controlling the caching placement and unit payment to each user. Assuming that the MNO can predict the responses of the users, a cost minimization problem is formulated, which is an NP-hard problem. By reformulating the problem into a non-negative submodular maximization problem over a matroid constraint, a local search algorithm is developed, which achieves $\left(\frac{1}{4+\epsilon}\right)$-approximation with polynomial runtime. Simulation results show that the local search algorithm provides near-to-optimal performance. Meanwhile, we also demonstrate that the incentive-based mobility-aware caching strategy efficiently offloads the cellular traffic to D2D links and reduces the overall cost.

\section{System Model}

\subsection{User Mobility Model}

As shown in Fig. \ref{mobility}, we consider an MNO, which provides wireless service to mobile users in a macro cell through a macro cell BS (MBS). There are $N_u$ users in the cell, whose index set is denoted as $\mathcal{D}=\{1, 2, ... , N_u\}$. A widely used inter-contact model is adopted to  capture the user mobility pattern, where the timeline of each user pair is divided into \emph{contact times} and \emph{inter-contact times}. The contact times denote the times when two mobile users are within the transmission range and are thus able to exchange data, while the inter-contact times denote the times between two consecutive contact times. As commonly assumed \cite{intercontactmodel}, the arrival of the contact times between user $i \in \mathcal{D} $ and $j \in \mathcal{D}$ is modeled as a Poisson process with intensity $\lambda_{i,j}$, and $\lambda_{i,i}$ is regarded as $+\infty$. Meanwhile, the timelines between different user pairs are assumed to be independent.

\subsection{Caching and File Delivery Model}

\begin{figure}[!t]
  \centering
  \includegraphics[width=3in]{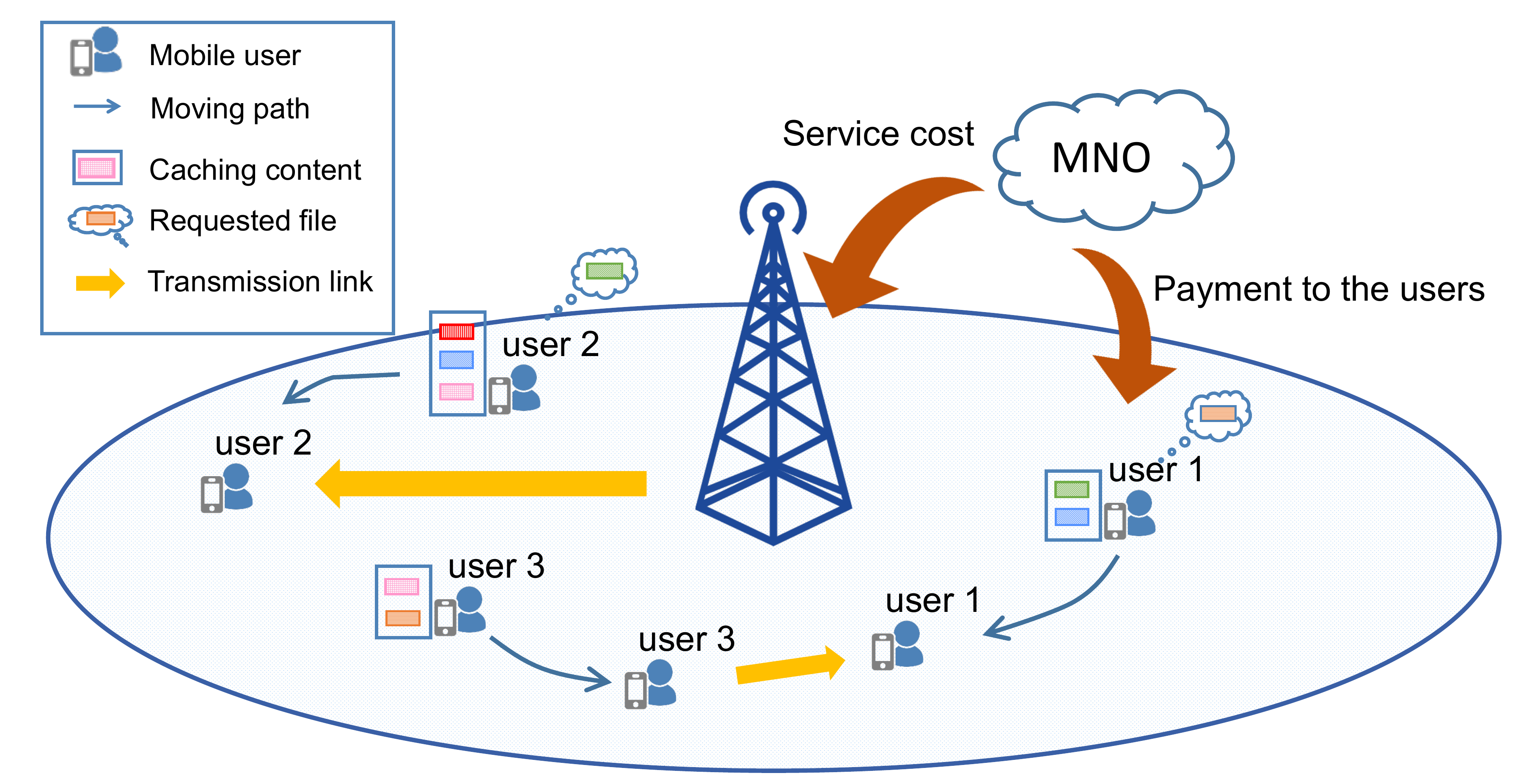}
  \caption{An illustration of the cache-assisted D2D network.}
  \label{mobility}
\end{figure}
A library with $N_f$ files is considered, whose index set is denoted as $\mathcal{F}=\{1, 2,..., N_f \}$. The size of each file is assumed to be $s$ bytes, while the available storage of user $i \in \mathcal{D}$ is $c_i$ bytes. Each file is assumed to be either fully cached or not stored at all at each user device. Specifically, $x_{i,f}=1$, where $ i \in \mathcal{D}$ and $ f \in \mathcal{F}$, means file $f$ is cached at user $i$; otherwise, $x_{i,f}=0$. File $f \in \mathcal{F}$ is requested by user $i \in \mathcal{D}$ with probability $p^r_{i,f}$, where $\sum _{f \in \mathcal{F}} p^r_{i,f}=1$. A user $i \in \mathcal{D}$ requesting file $f \in \mathcal{F}$ can be served via D2D links if it caches file $f$ or encounters other users who cache file $f$ within a certain delay time after the request is generated, denoted as $T^d$. Otherwise, it will be served by the MNO through a cellular link. Offloading user requests to D2D links can reduce the cellular load, and thus, lower the service cost. Meanwhile, it can also facilitate D2D links and improve the spectrum efficiency. An example is shown in Fig. \ref{mobility}. In this case, user 1 requests a file, and it encounters user 3 caching the requested file within the delay time. Then, it is served by user 3 via a D2D link. However, user 2 does not encounter a user caching its requested file, and thus, it is served via a cellular link.

\subsection{Pricing Model}

The utility of user $i$ using $w_i$ proportion of its available storage for its own needs, e.g., taking pictures and downloading apps, is denoted as $u^c_i(w_i)$. According to the law of diminishing returns, the same as \cite{poularakis2016mobile}, we consider the utility  $u^c_i(w_i)$ as a logarithm function given as
\begin{equation} \label{uc}
u^c_i(w_i)=a_i \ln (b_i w_i),
\end{equation} 
where $a_i$ and $b_i$ are parameters for characterizing the user utility. The cost of user $i$ for using $v_i$ bytes for caching is assumed to be linear of $v_i$, and the unit cost is $\rho_i$.  Meanwhile, if user $i$ shares $v_i$ bytes for caching, the MNO will pay $r_i v_i$ for it, where $r_i$ is the unit payment.
\section{Incentive Mechanism Design}
In this section, a Stackelberg game will be formulated. In a Stackelberg game, a leader and several followers compete sequentially on certain recourses. The leader goes first, and then the followers make decisions based on the choice of the leader \cite{myerson2013game}. In this paper, to offload traffic to D2D links and reduce the service cost, the MNO provides payment to the mobile users to motivate device caching. Accordingly, the MNO plays as a leader aiming at minimizing the overall cost, including the service cost and the payment. Considering the selfish nature, each mobile user plays as a follower aiming at maximizing its own utility. Due to the high overhead, it is impractical for each user to acquire the global information, including the pairwise mobility information and the caching placement at other user devices. Thus, we assume that the MNO designs the caching placement and determines the payment based on the global information. Meanwhile, each user determines its cache storage based on its own local information.

\subsection{User Subgame}
The users are regarded as the followers in the Stackelberg game. The overall utility of user $i \in \mathcal{D}$ includes the utility of the storage for its own needs, the cost for caching, and the payment from the MNO, given as
$
U_i(v_i)=u^c_i \left(1-\frac{v_i}{c_i}\right)-\rho_i v_i + r_i v_i.
$
Each user aims at maximizing its own utility, and thus, the user subgame is
\begin{equation}
v^{\star}_i= \arg \max \limits_{v_i \in [0,c_i]} U_i(v_i),
\end{equation}
where $i \in \mathcal{D}$. By finding the stationary point, we can get the optimal cache storage for user $i \in \mathcal{D}$ as 
\begin{equation}
v^{\star}_i=
\begin{cases}
c_i-\frac{a_i}{r_i - \rho_i} & \text{if } r_i > \rho_i, \\
0     & \text{if } r_i \le \rho_i.
\end{cases}
\end{equation}

\subsection{Cost Minimization Problem}
The MNO is regarded as the leader in the Stackelberg game. The overall cost of the MNO includes the cost to serve the mobile users via cellular links and the payment to the users. According to \cite{courcoubetis2003pricing}, we assume that the service cost is a function of the proportion of requested files served via cellular links, denoted as $Q(P^c)$, where $P^c$ is the proportion of requested files served via cellular links and $Q(\cdot)$ is assumed to be an increasing convex function. In the following, we will first find an expression for $P^c$.

Let $T^e_{i,j}$, where $i,j \in \mathcal{D}$, denote the time between user $i$ requesting a file and user $i$ encountering user $j$, and then the delay time of user $i$ when requesting file $f$ is given as
$
D_{i,f}= \min \limits_{j \in \mathcal{D} \text{ and } x_{j,f}=1} T^e_{i,j}.
$
Based on the memoryless property of the Poisson process, $T^e_{i,j}$ follows an exponential distribution with parameter $\lambda_{i,j}$. Since the inter-contact times for different pairs of users are assumed to be independent, $D_{i,f}$ is an exponential distributed random variable with parameter
$
\Lambda_{i,f}=\sum \limits_{j \in \mathcal{D} \text{ and } x_{j,f}=1} \lambda_{i,j}=\sum \limits_{j \in \mathcal{D} } x_{j,f} \lambda_{i,j}.
$
Since user $i$ can get the requested file via D2D links when it stores the file in its own cache or the delay time is within time $T^d$, the probability that this user needs to download the requested file via a cellular link is
$
P_i=\sum \limits_{f \in \mathcal{F}} p_f \exp \left(  -\sum \limits_{j \in \mathcal{D} } x_{j,f} T^d \lambda_{i,j}\right) .
$
Thus, the proportion of requested files served via cellular links is given as $P^c=\frac{1}{N_u} \sum\limits_{i \in \mathcal{D}} P_i$. 
Then, the cost minimization problem for the MNO is 
\begin{align}
\min \limits_{x_{i,f},r_i}  & \quad Q(P^c)+\sum \limits_{i \in \mathcal{D}} r_i v_i,  \label{prob_MNOgame} \\
\text{s.t. } & \quad s \sum \limits_{f \in \mathcal{F}} x_{i,f} \le v_i, i \in \mathcal{D} \text{ and } x_{i,f} \in \{0,1\}, \tag{\ref{prob_MNOgame}a}
\end{align}
where constraint (\ref{prob_MNOgame}a) implies the limited cache storage. Assuming that the utility function of the storage for its own needs and the unit caching cost is stable in a period, the MNO can use learning algorithms, e.g., regression analysis \cite{bates1988nonlinear}, to estimate the parameters $a_i$, $b_i$, $\rho_i$ and $c_i$, where $i \in \mathcal{D}$. Accordingly, the MNO can predict the responses of the mobile users. Specifically, if the MNO needs cache storage as $v_i \in [0,c_i]$ bytes of user $i$, the unit payment to user $i$ should be
\begin{equation} \label{price}
r_i=\frac{a_i}{c_i-v_i}+\rho_i,
\end{equation} 
and thus the payment to user $i$ is
\begin{equation}
C^P_i(v_i)=\left(\frac{a_i}{c_i-v_i}+\rho_i \right)v_i.
\end{equation} 
Since the required cache storage of user $i \in \mathcal{D}$ is $s \sum \limits_{f \in \mathcal{F}} x_{i,f} $, the cost minimization problem considering the response of each user is 
\begin{align}
\min \limits_{x_{i,f}} & \quad Q(P^c)+\sum \limits_{i \in \mathcal{D}} C_i^P\left( s \sum \limits_{f \in \mathcal{F}} x_{i,f}   \right), \label{prob_MNO}\\
\text{s.t. } & \quad s \sum \limits_{f \in \mathcal{F}} x_{i,f} \le c_i, i \in \mathcal{D} \text{ and } x_{i,f} \in \{0,1\} .  \tag{\ref{prob_MNO}a}
\end{align}
Considering that problem (\ref{prob_MNO}) is an NP-hard problem, in the following, we will provide an effective sub-optimal algorithm to solve it.

\section{Mobility-Aware Caching Strategy}
In this section, we first reformulate the cost minimization problem as a non-negative submodular maximization problem over a matroid constraint. Then, a local search algorithm is adopted to solve the problem, which achieves a $\left(\frac{1}{4+\epsilon}\right)$-approximation.
\subsection{Problem Reformulation}
Submodular maximization problems belong to an important kind of combinational optimization problems, which have been widely investigated \cite{bordeaux2014tractability, lee2010maximizing}. Some recent works developed caching strategies based on submodular maximization over a matroid constraint \cite{mobilitycaching,femtocachinginfocom}. 

We use $y_{j,f}$, where $j \in \mathcal{D}$ and $ f \in \mathcal{F}$, to denote that file $f$ is cached at user $j$, and define the ground set as $S=\{y_{j,f}| j \in \mathcal{D} \text{ and } f \in \mathcal{F}\}$. Then, each caching placement can be represented as a subset of $S$. Specifically, for a caching placement set $A \subseteq S$, the element $y_{j,f} \in A $ means that user $j$ caches file $f$, i.e., $x_{j,f}=1$, while $y_{j,f} \notin A $ means that user $j$ does not cache file $f$, i.e., $x_{j,f}=0$.
Denote $S_i=\{y_{i,f}| f \in \mathcal{F}\}$, which contains all the files that may be cached at user $i$, and then the constraint (\ref{prob_MNO}a) can be rewritten as a matroid constraint, as shown in Lemma \ref{matroidc}.
\begin{lem} \label{matroidc}
Let $c^n_i=\max \left( 0, \left\lceil \frac{c_i}{s}-1 \right\rceil \right)$, where $i \in \mathcal{D}$. Constraint (\ref{prob_MNO}a) can be rewritten as a matroid constraint, i.e., $Y \in \mathcal{I}$, where
\begin{equation} \label{I}
\mathcal{I}=\left\{A \subseteq S \big| |A \cap S_i| \le c^n_i, \forall i \in \mathcal{D} \right\}.
\end{equation}
\end{lem}
\begin{proof}
Based on Eq. (\ref{price}), if the required cache storage of user $i \in \mathcal{D}$ equals $c_i$, i.e., $s \sum \limits_{f \in \mathcal{F}} x_{i,f}=c_i$, the payment to user $i$ goes to $\infty$. Since the objective is to minimize the overall cost, it is impossible to have $s \sum \limits_{f \in \mathcal{F}} x_{i,f}=c_i$, which also fits the practical scenario that the mobile user does not share all the available storage for caching. Considering that  the number of files cached at user $i$ is an integer, constraint (\ref{prob_MNO}a) can be rewritten as (\ref{I}). Then, the tuple $\mathcal{M}=(S,\mathcal{I})$ is a kind of typical matroids, i,e., the partition matroid \cite{submodular2}.
\end{proof}
Then, the cost minimization problem (\ref{prob_MNO}) is equivalent to the following maximization problem
\begin{align} \label{prob_sub}
\max \limits_{Y \in \mathcal{I}} \quad g(Y)=\theta - Q \left(P^c(Y) \right) - \sum \limits_{i \in \mathcal{D}} C^A_i(|Y \cap S_i|),
\end{align}
where the proportion of requested files served via cellular links is rewritten as a set function, given as
\begin{equation}
P^c(Y)=\frac{1}{N_u} \sum \limits_{i \in \mathcal{D}} \sum \limits_{f \in \mathcal{F}} p_f \exp \left(  -\sum \limits_{j \in \mathcal{D}, y_{j,f} \in Y  } T^d \lambda_{i,j}\right),
\end{equation} 
the payment to user $i \in \mathcal{D}$ is rewritten as
\begin{align} \label{pay}
& C^A_i(|Y \cap S_i|)= \notag \\
& \left\{
\begin{array}{ll}
\vspace{0.05in} C^P_i\left(s |Y \cap S_i|\right) & \text{if } |Y \cap S_i| \le c_i^n, \\
\left[  C^P_i  \left( s c^n_i\right) - C^P_i\left(s \cdot \max(0,c^n_i -1)\right)\right]  & \multirow{2}{*}{if $|Y \cap S_i| > c_i^n$,} \\
	\times \left( |Y \cap S_i|- c^n_i\right) + C^P_i\left(s c_i^n\right)&
\end{array}
\right.
\end{align}
and $\theta=Q(1)+ \sum \limits_{i \in \mathcal{D}} C^A_i \left( N_f \right)$ is a constant to guarantee that $g(Y)$ is a non-negative function. Note that, rewriting the payment as in (\ref{pay}) is to cover the domain of the ground set $S$, including the subsets of $S$ which is not in $\mathcal{I}$. Then, Lemma \ref{submodular} proves the submodularity of $g(Y)$ in (\ref{prob_sub}).
\begin{lem} \label{submodular}
	The objective function in (\ref{prob_sub}), i.e., $g(Y)$, is a non-negative submodular set function.
\end{lem}
\begin{proof}
	Due to space limitation, the proof is omitted. 
\end{proof}
Based on Lemmas \ref{matroidc} and \ref{submodular}, problem (\ref{prob_sub}) is to maximize a non-negative submodular function over a matroid constraint.
\subsection{Local Search Algorithm}
To solve a non-negative submodular maximization problem over a matroid constraint, the local search algorithm provides a sub-optimal solution, which is at least $\left(\frac{1}{4+\epsilon}\right)$ of the optimal solution \cite{lee2010maximizing}. Fractional local search algorithms can provide solutions with higher approximation ratios. However, considering that the size of the ground set is $\left(N_u  N_f \right)$, the fractional local search algorithms are impractical due to the high computation complexity. Furthermore, in the simulation, we will show that the performance of the local search algorithm is close to that of the optimal one.
\begin{algorithm}[!h]
\caption{The Local Search Procedure}
\label{alg:lsp}
\begin{algorithmic}[1]
\REQUIRE{$\left( V, \mathcal{I}, g \right)$ including a set $V$, a matroid constraint $\mathcal{I}$, a non-negative submodular set function $g(Y)$}
\ENSURE{A solution $Y^\star$ corresponding to $\arg\max \limits_{Y \in \mathcal{I}, Y \subseteq V} g(Y)$}
\STATE {Select $y_{j^{\star},f^{\star}}=\arg \max \limits_{\{y_{j,f}\} \in \mathcal{I} \text{ and } y_{j,f} \in V } {g(\{y_{j,f}\} )}$ and set $Y=\{y_{j^\star,f^\star}\} $.}
\WHILE{any of the following local improvements applies }
\STATE {Add operation:}
\IF {$\exists y_{j,f} \in V \backslash Y $ such that  $Y \cup \{y_{j,f}\} \in \mathcal{I}$ and $g \left( Y\cup\{y_{j,f}\} \right) \ge \left( 1+ \frac{\epsilon}{N_u^4 N_f^4}\right) g \left( Y \right) $}
\STATE {Set $Y=Y \cup \{y_{j,f}\} $.}
\ENDIF
\STATE {Delete operation:}
\IF {$\exists y_{j,f} \in Y $ such that  $g \left( Y \backslash \{y_{j,f}\} \right) \ge \left( 1+ \frac{\epsilon}{N_u^4 N_f^4}\right) g \left( Y \right) $}
\STATE {Set $Y=Y \backslash \{y_{j,f}\} $.}
\ENDIF
\STATE {Swap operation:}
\IF {$\exists y_{j,f} \in V\backslash Y $ and $y_{j',f'} \in Y $ such that  $Y\backslash \{y_{j',f'}\} \cup \{y_{j,f}\} \in \mathcal{I}$ and $g \left(Y\backslash \{y_{j',f'}\} \cup \{y_{j,f}\}  \right) \ge \left( 1+ \frac{\epsilon}{N_u^4 N_f^4}\right) g \left( Y \right) $}
\STATE {Set $Y=Y\backslash \{y_{j',f'}\} \cup \{y_{j,f}\} $.}
\ENDIF
\ENDWHILE
\end{algorithmic}
\end{algorithm}

\begin{algorithm}[!h]
\caption{The Local Search Algorithm}
\label{alg:lsa}
\begin{algorithmic}[1]
\REQUIRE{$\left( S, \mathcal{I}, g \right)$ including a ground set $S$, a matroid constraint $\mathcal{I}$, a non-negative submodular function $g(Y)$.}
\ENSURE{A solution $Y^\star$ corresponding to $\arg\max \limits_{Y \in \mathcal{I}, Y \subseteq S} g(Y)$}
\STATE {Set $V_1=S$.}
\FOR{ $i =1,2$ }
\STATE{Apply the Algorithm \ref{alg:lsp} with input $\left( V_i, \mathcal{I}, g \right)$ and get the output solution $Y^\star_i$.}
\STATE{Set $V_{i+1}=V_i \backslash Y^\star_i$.}
\ENDFOR
\STATE{Return the solution $Y^\star=\arg \max \left( g(Y^\star_1), g(Y^\star_2) \right)$.}
\end{algorithmic}
\end{algorithm}

As listed in Algorithm \ref{alg:lsp}, the local search procedure finds a solution of problem (\ref{prob_sub}) on a set, i.e., $V$. There are three local improvement operations, including adding an element, deleting an element and swapping a pair of elements. The procedure starts with selecting an element $y_{j,f} \in V$ which maximizes  $g(\{y_{j,f}\} )$ subjecting to $\{ y_{j,f} \} \in \mathcal{I}$, denoted as $y_{j^\star,f^\star}$, and initializes the solution set as  $Y=\{y_{j^\star,f^\star}\}$. Then, the solution is improved via the three operations, until none of them can improve the objective value by more than $\frac{\epsilon}{N_u^4 N_f^4}$. The local search algorithm in Algorithm \ref{alg:lsa} is a $\left( \frac{1}{4+\epsilon} \right)$-approximate algorithm of problem (\ref{prob_sub}). It first applies Algorithm \ref{alg:lsp} on the ground set $S$ and gets a corresponding solution $Y^\star_1$. Then, Algorithm  \ref{alg:lsp} is applied on the set $S \backslash Y^\star_1$ and a corresponding solution $Y^\star_2$ is obtained. Finally, the approximate solution $Y^\star$ is obtained as the optimal one of $Y^\star_1$ and $Y^\star_2$.
It is proved in  \cite{lee2010maximizing} that, if $\frac{1}{\epsilon}$ is no more than a polynomial of $\left(N_u N_f\right)$,  Algorithm \ref{alg:lsp} can run in polynomial time.
\section{Simulation Results}
In this section, we provide simulation results to validate the effectiveness of the local search algorithm and the incentive-based mobility-aware caching strategy. The following four strategies are compared:
\begin{enumerate}
\item{Optimal mobility-aware caching strategy:} it applies the optimal solution of problem (\ref{prob_MNO}) obtained by a dynamic programming (DP) algorithm. The details of the algorithm are omitted due to space limitation.
\item{Sub-optimal mobility-aware caching strategy:} it applies the sub-optimal solution of problem (\ref{prob_MNO}) obtained by the local search algorithm in Section IV, where $\epsilon=0.01$. 
\item{Popular caching strategy:} each user caches the most popular files.
\item{Random caching strategy:} the probability of each user to cache a particular file is proportional to the file request probability.
\end{enumerate}
In both the popular and random caching strategies, the cache storage of each mobile user is assumed to be the same, and the optimal cache storage is obtained by line search. The file request probability follows a Zipf distribution with parameter $\gamma_r$, i.e.,  $p^r_{i,f}=\frac{f^{-\gamma_r}}{\sum \limits_{k \in \mathcal{F}} k^{-\gamma_r}}$, $i \in \mathcal{D}$.  The requests are more concentrated on the popular files with a larger $\gamma_r$. We assume the service cost function $Q(\cdot)$ to be a linear function, the slop of which is set as $0.01$ \$/MB based on the charges of AT\&T's and Verizon's plans \cite{joe2013offering}. The size of each file is assumed to be $200 \text{ MB}$, and each user requests one file per day on average. Thus, we have $Q(P^c)=0.01 \times 200 N_u P^c$, denoting the service cost per day. The overall cost is normalized by the service cost without caching, i.e., $Q(1)$. The available storage of each user is assumed to be $1$ GB. Meanwhile, the parameters $\rho_i$, $b_i$ and $a_i$, $i \in \mathcal{D}$, are set as $0$, $100$ and $0.015/\ln(100)$ \$\footnote{According to Google Drive, the price of online storage is $1.99$ \$ per month for $100$ GB. The price of the storage hardware on the mobile phone is roughly $20$ times higher than that of the online HDD storage, which is estimated by the storage price of iPhone and the price of HDD on Amazon. The price of the storage on the mobile phones is therefore approximately $0.015$ \$/GB per day. Assuming $b_i=100$ and $u_i^c(1)=0.015$ \$, we can get the value of $a_i$ based on Eq. (\ref{uc}).}, respectively.

\begin{figure}[t]
	\centering
	\includegraphics[width=2.6in]{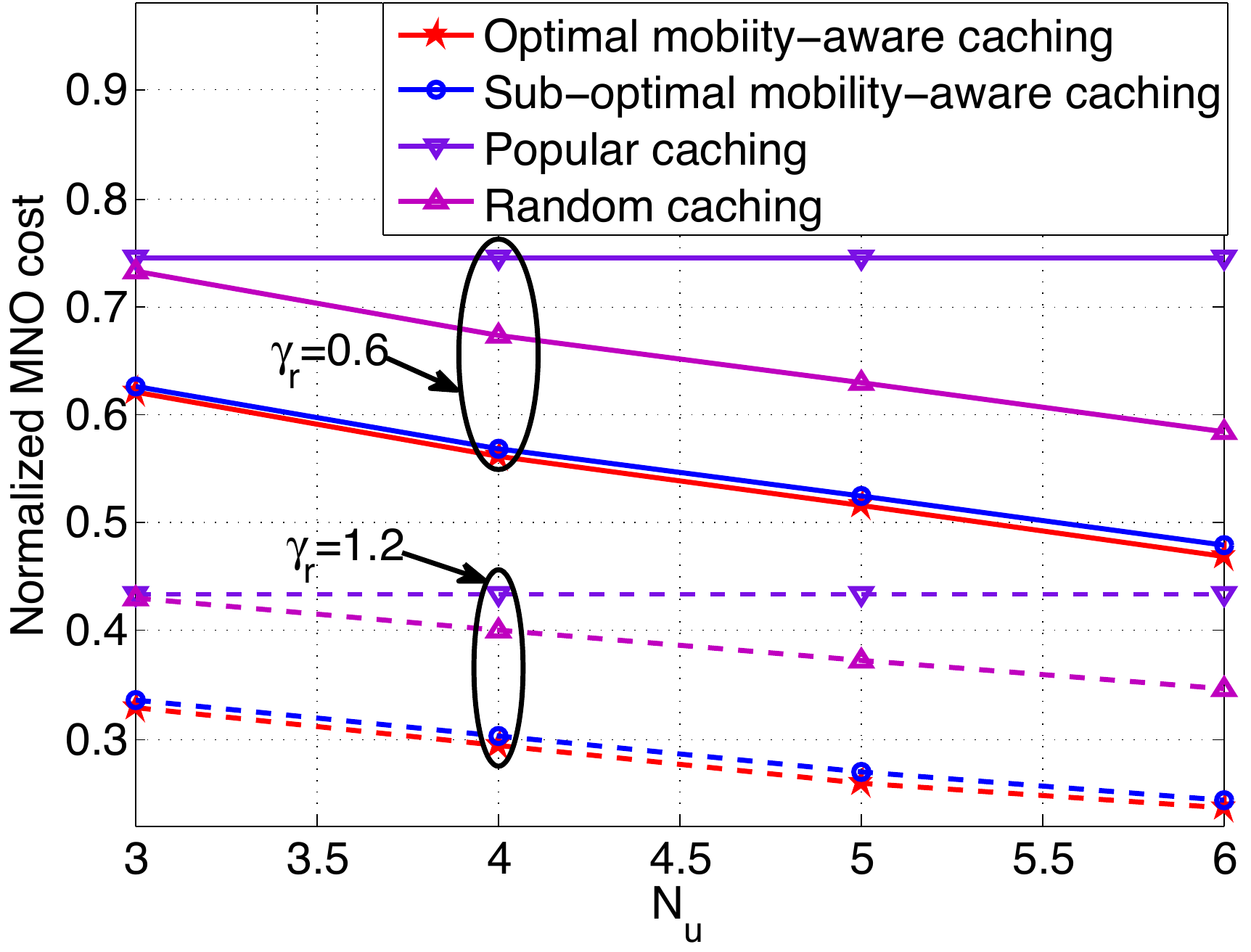}
	\caption{Comparison of different caching strategies with $N_{file}=50$, and $T^d=300s$.}
	\label{fig1}
\end{figure}
\begin{figure}[t]
	\centering
	\includegraphics[width=2.7in]{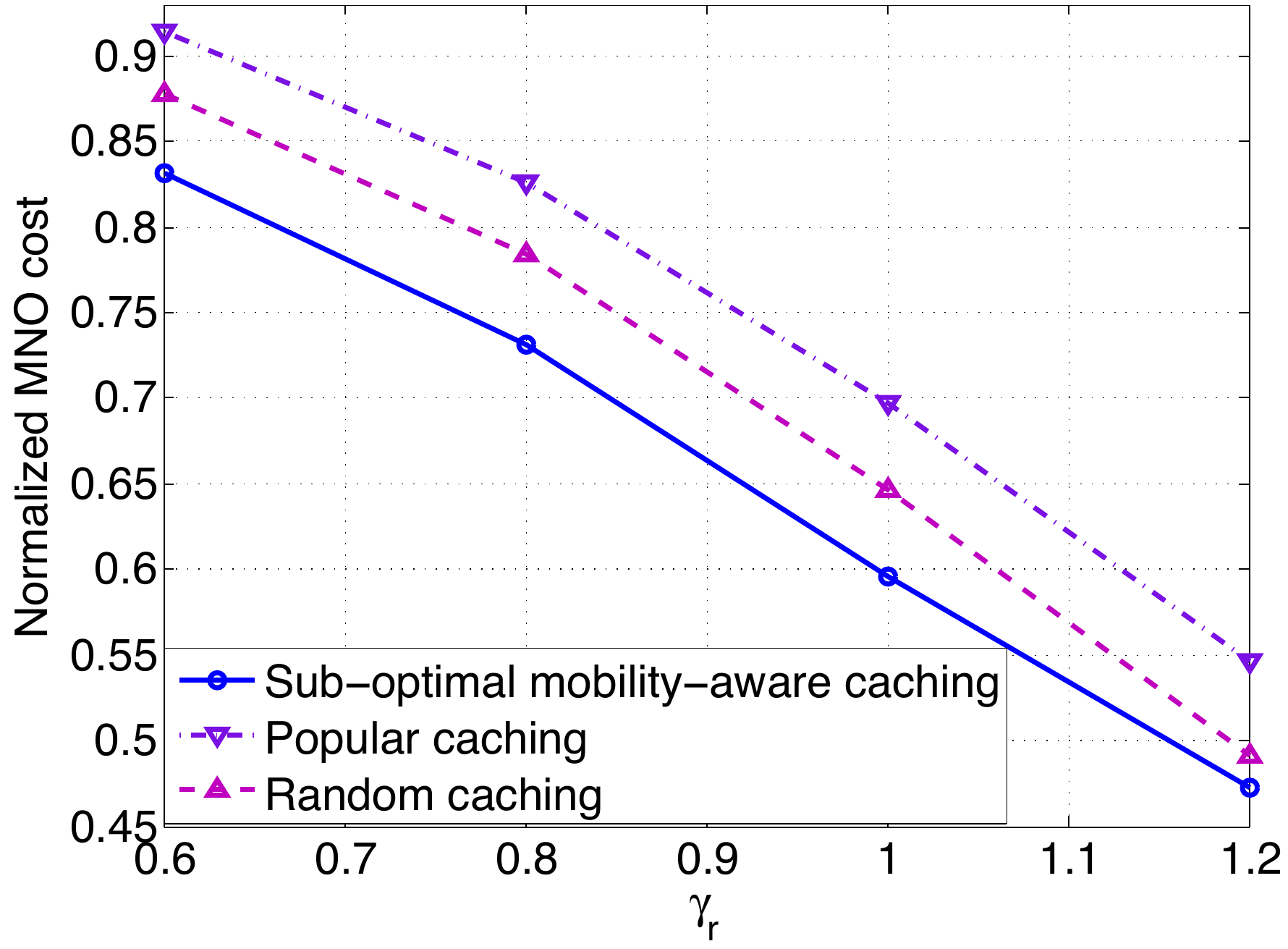}
	\caption{Comparison of different caching strategies with $N_{file}=500$, and $T^d=300s$.}
	\label{fig2}
\end{figure}

In Fig. \ref{fig1}, the inter-contact parameters $\lambda_{i,j}$, where $i \in \mathcal{D}$ and $j \in \mathcal{D} \backslash \{i\}$, are generated according to a Gamma distribution, i.e., $\Gamma(4.43, 1/1088)$ \cite{aggregate_real}. By comparing the optimal and sub-optimal mobility-aware caching strategies, we see that the local search algorithm provides near optimal performance. Moreover, the mobility-aware caching strategies outperform both popular and random caching strategies, since the latter ones do not make good use of the user mobility information and the prediction of user response. 

Next we use a real-life data set collected by Chaintreau \emph{et. al.} in \cite{intercontactmodel} to evaluate the performance of the proposed incentive-based mobility-aware caching strategy. In this data set, the contact times among 78 participants are recorded. Same as \cite{mobilitycaching}, we estimate the parameters $\lambda_{i,j}$ by the average contact rates during the daytime of the first day and design the caching strategies. The performance in the second daytime is shown in Fig. \ref{fig2}. It is shown that the proposed mobility-aware caching strategy can better reduce the overall cost than the popular and random caching strategies.
\section{Conclusions}
In this paper, we designed an incentive mechanism for a D2D caching network while taking advantage of the user mobility pattern. A Stackelberg game was formulated, where each user plays as a follower and the MNO plays as the leader. We developed an effective local search algorithm to solve the cost minimization problem at the MNO. Simulation results showed that the local search algorithm provides near optimal performance and also validated the effectiveness of the proposed incentive-based mobility-aware caching strategy. For future works, it would be interesting to develop online caching update strategies considering time-varying user preferences and utilities.






\bibliographystyle{IEEEtran}
\bibliography{IEEEabrv,report}
\end{document}